\documentclass[a4paper,reqno,12pt]{amsart} 
\pdfoutput=1
\makeatletter
\def\subsection{\@startsection{subsection}{2}%
  \z@{.5\linespacing\@plus.7\linespacing}{.3\linespacing}%
  {\normalfont\bfseries}}
\makeatother

\usepackage[english]{babel}
\usepackage[utf8]{inputenc}
\usepackage[T1]{fontenc}
\usepackage{eucal}
\usepackage{xparse}
\usepackage{microtype}
\usepackage{csquotes}
\usepackage{mathtools}
\usepackage{amssymb,amsfonts}

\usepackage[backend=biber,giveninits=true,maxnames=5,style=trad-plain]{biblatex} 
\AtEveryBibitem{%
    \clearlist{language}
}

\usepackage[hyperindex,breaklinks]{hyperref}
\usepackage{cleveref} 

\DeclarePairedDelimiter{\parens}{\lparen}{\rparen}

\DeclarePairedDelimiter{\set}{\{}{\}}
\DeclarePairedDelimiter{\sqbracks}{[}{]}

\newcommand*{\RR}{\mathbb{R}}

\DeclareMathOperator{\Id}{Id}
\newcommand*{\comp}{\circ} %

\DeclareMathOperator{\supp}{supp} %

\DeclareMathOperator{\Cont}{\mathcal{C}}

\newcommand{\dd}{\mathrm{d}}
\NewDocumentCommand{\dv}{mm}{\frac{\dd #1}{\dd #2}}
\NewDocumentCommand{\pdv}{mm}{\frac{\partial #1}{\partial #2}}

\NewDocumentCommand{\orth}{om}{{#2}^{\bot\IfValueT{#1}{_#1}}}
\NewDocumentCommand{\orthL}{om}{\prescript{\bot\IfValueT{#1}{_#1}}{}{#2}}

\newcommand*{\lieBr}[1]{\sqbracks{#1}}
\newcommand*{\lieD}[1]{\mathcal{L}_{#1}}

\DeclareMathOperator{\Ad}{Ad}
\newcommand{\contr}[1]{\iota_{#1}}
\newcommand*{\Forms}{\Omega}
\newcommand*{\VecFields}{\mathfrak{X}}

\theoremstyle{plain}
\newtheorem{theorem}{Theorem}
\newtheorem*{theorem*}{Theorem}
\newtheorem{lemma}[theorem]{Lemma}
\newtheorem*{lemma*}{Lemma}
\newtheorem{proposition}[theorem]{Proposition}
\newtheorem*{proposition*}{Proposition}

\newtheorem*{cor*}{Corollary}

\theoremstyle{definition}
\newtheorem{definition}{Definition}
\newtheorem*{definition*}{Definition}

\newtheorem*{example*}{Example}

\crefname{theorem}{Theorem}{Theorems}
\Crefname{theorem}{Theorem}{Theorems}
\crefname{lemma}{Lemma}{Lemmas}
\Crefname{lemma}{Lemma}{Lemmas}
\crefname{proposition}{Proposition}{Propositions}
\Crefname{Prop}{Proposition}{Propositions}
\crefname{cor}{Corollary}{Corollaries}
\Crefname{cor}{Corollary}{Corollaries}
\crefname{definition}{Definition}{Definitions}
\Crefname{definition}{Definition}{Definitions}
\crefname{example}{Example}{Examples}
\Crefname{example}{Example}{Examples}
\crefname{theorem}{Theorem}{Theorems}

\newtheorem*{exercise*}{Exercise}
\crefname{exercise}{exercise}{exercises}
\Crefname{exercise}{Exercise}{Exercises}  

\theoremstyle{remark}
\newtheorem{remarkx}{Remark}
\newtheorem*{remarkx*}{Remark}
\crefname{remark}{Remark}{Remarks}
\Crefname{remark}{Remark}{Remarks}
\newenvironment{remark}
  {\pushQED{\qed}\remarkx}
  {\popQED\endremarkx}
\newenvironment{remark*}
  {\pushQED{\qed}\remarkx*}
  {\popQED\endremarkx*}

    \newcommand{\Reeb}{\mathcal{R}}

    \newcommand{\lsharp}{\sharp_\Lambda}

    \DeclarePairedDelimiter{\jacBr}{\lbrace}{\rbrace}
    
    \newcommand{\clift}[1]{{#1}^C}
    \newcommand{\vlift}[1]{{#1}^V}
    \newcommand{\rclift}[1]{\bar{{#1}}^C}
    \newcommand{\rvlift}[1]{\bar{{#1}}^V}
    \newcommand*{\rVecFields}{\bar{\mathfrak{X}}}

    \title{Infinitesimal symmetries in Contact Hamiltonian systems}

   \author[M. de León]{Manuel de León}
   \address{Manuel de Le\'on: Instituto de Ciencias Matem\'aticas (CSIC-UAM-UC3M-UCM),
   c\textbackslash Nicol\'as Cabrera, 13-15, Campus Cantoblanco, UAM
   28049 Madrid, Spain \newline
   and \newline
   Real Academia de Ciencias Exactas, Físicas y Naturales, c\textbackslash de Valverde,
   22, 28004 Madrid, Spain
   } \email{mdeleon@icmat.es}

   \author[M. Lainz Valcázar]{Manuel Lainz Valcázar}
   \address{Manuel Laínz Valcázar:
   Instituto de Ciencias Matem\'aticas (CSIC-UAM-UC3M-UCM),
   c$\backslash$ Nicol\'as Cabrera, 13-15, Campus Cantoblanco, UAM
   28049 Madrid, Spain} \email{manuel.lainz@icmat.es}

   \date{\today}
\begin{document}

\begin{abstract}
    In this paper, we extend the well-known Noether theorem for Lagrangian systems to contact Lagrangian systems. We introduce a classification of infinitesimal symmetries and obtain the corresponding dissipated quantities. We notice that in contact dynamics, the existence of infinitesimal symmetries does not produce conserved quantities, but functions that dissipate at the same rate than the energy; so, the corresponding quotients are true conserved quantities.
\end{abstract}
\maketitle

\section{Introduction}

Noether theorem is one of the most relevant results relating symmetries of a Lagrangian system and conserved quantities of the corresponding Euler-Lagrange equations. In the simplest view, the existence of a \emph{cyclic} coordinate implies the conservation of the corresponding momentum. Indeed, if $L=L(q^i, \dot{q}^i)$ does not depend on the coordinate $q^j$, then, using the Euler-Lagrange equation
\begin{equation}
    \dv{}{t} \parens*{\pdv{L}{\dot{q}^j}} - \pdv{L}{q^j} = 0,
\end{equation}
we deduce that (see~\cite{Arnold1997})
\begin{equation}
    \dot{p}_j =  \dv{}{t} \parens*{\pdv{L}{\dot{q}^j}} = 0.
\end{equation}

Noether theorem can be described on a geometric framework~\cite{Carinena1992,Carinena1991,Carinena1989,Carinena1988,Sarlet1987,Cantrijn1980,Sarlet1983,deLeon1994,deLeon1994a,Cicogna1992,Prince1983,Prince1985,Crampin1983,Carinena1989a,Aldaya1978,Aldaya1980}. In that framework, $L$ is a function on the tangent bundle $TQ$ of the configuration manifold $Q$ and $X$ be a vector field on $Q$. We denote by $\clift{X}$ and $\vlift{X}$ the complete and vertical lifts of $X$ to the tangent bundle $TQ$. Then (see~\cite{deLeon2011}):

\begin{theorem}[Noether]\label{thm:noether_intro}
    $\clift{X}(L) = 0$ if and only if $\vlift{X}(L)$ is a conserved quantity.
\end{theorem}

Here we are using the symplectic formulation of Lagrangian mechanics. Hence, $L$ defines a symplectic form
\begin{equation}
    \omega_L = - \dd \alpha_L, \quad \alpha_L = S^* (\dd L)
\end{equation}
on $TQ$, where $L$ is assumed to be regular, $S$ is the canonical vertical endomorphism on $TQ$, $S^*$ is the adjoint operator, and the dynamics is obtained by the equation
\begin{equation}
    \contr{\xi_L} \omega_L = \dd E_L,
\end{equation}
where $E_L = \Delta(L) - L$ is the associated energy and $\Delta$ is the canonical Liouville or Euler vector field on $TQ$. 
Indeed the projection to $Q$ of the integral curves of the second order differential equation $\xi_L$ are just the solutions of the Euler-Lagrange equations.

This approach has permitted a deep investigation on other possible infinitesimal symmetries, relating them with the corresponding conserved quantities. A first distinction with the Hamiltonian framework is that we can consider point-base symmetries and symmetries on the phase space of velocities.

The literature about this subject is indeed very extensive. See for example~\textcite{Cantrijn1980}, \textcite{Sarlet1983}, \textcite{Prince1983,Prince1985}, \textcite{Crampin1983}, \textcite{Marmo1986}, \textcite{Cicogna1992},~\textcite{Aldaya1978,Aldaya1980}, even with more general symmetries, \textcite{Sarlet1987}, or, for the time dependent case, Cariñena \emph{et al.}~\cite{Carinena1991,Carinena1989,Carinena1989a,Carinena1992}, or singular Lagrangian systems~\cite{Carinena1988}, and for higher order Lagrangian systems in~\textcite{deLeon1994a,deLeon1994,deLeon1995}.

The aim of the present paper is to extend this theory to the case of contact Lagrangian systems defined as follows. Let $L: TQ\times \RR \to \RR$ be a Lagrangian depending on an extra parameter $z \in \RR$, then the dynamics is obtained through the contact Hamiltonian system given by the Hamiltonian function $E_L = \Delta(L)-L$ and the contact form
\begin{equation}
    \eta_L = \dd z - \alpha_L,
\end{equation}
where $\alpha_L$ is the pullback to $TQ\times \RR$ of the form defined on $TQ$ for the symplectic case.

The first question is the following: we know that the contact Euler-Lagrange equations, also known as the Herglotz equations~\cite{Herglotz1930,deLeon2019},
\begin{equation}\label{eq:Herglotz_intro}
    \dv{}{t} \parens*{\pdv{L}{\dot{q}^i}} - \pdv{L}{q^i} = \pdv{L}{\dot{q}^i} \pdv{L}{z}
\end{equation}
can be obtained by a variational principle, called the Herglotz variational principle. The main difference with the usual Hamilton principle is that the action is defined by a non-autonomous ODE, instead of an integral.

In~\cite{Georgieva2003,Georgieva2011} a Noether theorem was proven for these Lagrangian systems. However, that proof is written in terms of calculus of variations. In addition, the proofs are not easy to follow.

The main result in the present paper is to provide a geometric framework for this theory in a similar vein to the symplectic case. This gives a geometric interpretation of the contact Noether theorem analogous to~\cref{thm:noether_intro}.

After this first result, we have proceeded to extend it to more general types of infinitesimal symmetries and computing the corresponding quantities.

A relevant comment here is that in contact Lagrangian systems we do not obtain conserved quantities, but quantities that dissipate at the same rate as the energy of the system $E_L$.  Those quantities will be called \emph{dissipated quantities}.
. This can already be seen in the case of a cyclic coordinate: if $L$ does not depend on $q^j$, then, by \cref{eq:Herglotz_intro},
\begin{align}
    \dot{p}_j &=\pdv{L}{z} p_j,
\intertext{or, along an integral curve,}
    p_j(t) &= p(0)  \int_0^t \pdv{L}{z} (t) \dd t .
\end{align}

As we know, we have
\begin{align}
    \dot{E_L} &=\pdv{L}{z} E_L,\\
\intertext{and then,}
    E_L(t) &= E_L(0)  \int_0^t \pdv{L}{z} (t) \dd t.
\end{align}
Note that, assuming that $E_L$ is nonzero, then $p_j / E_L$ is a conserved quantity.

We note that in~\cite{Gaset2019}, the authors describe the concept of infinitesimal symmetries and dissipated quantities on the Hamiltonian framework. Their results and definitions are particular cases of ours.

The paper is structured as follows. In Section~2 we study the relationship of infinitesimal symmetries and dissipated quantities on the general context of contact Hamiltonian systems. In Section~3, we study the specific case of contact Lagrangian systems, considering infinitesimal symmetries in increasing order of generality.
Indeed, we introduce infinitesimal symmetries of the Lagrangian, generalized infinitesimal symmetries of the Lagrangian, Noether symmetries and Lie symmetries.
A relevant point is that we can consider symmetries based on $Q$ or on $Q\times \RR$. The corresponding dissipated quantities are also obtained.
In Sections~4 and~5 we consider the case of a Lie group of symmetries and the corresponding momentum map.

\section{Symmetries and conserved quantities in contact Hamiltonian systems}
We recall some results in contact geometry. The detailed proofs can be found in~\cite{deLeon2018}.
 Let $(M,\eta)$ be a contact manifold. This means that $M$ is a $(2n+1)$-dimensional manifold and $\eta \wedge (\dd\eta)^n$ is a volume form.
Then, there exist a unique vector field $\Reeb$ (the so-called Reeb vector field) such that
\begin{equation}
	\contr{\Reeb}  \dd \eta = 0, \quad  \contr{\Reeb}\, \eta = 1.
\end{equation}

There is a Darboux theorem for contact manifolds so that around each point in $M$ one can find local coordinates 
(called Darboux coordinates) $(q^i, p_i, z)$ 
such that
\begin{equation}
	 \eta = \dd z - p_i \, \dd q^i.
\end{equation}

In Darboux coordinates we have
\begin{equation}
	\Reeb = \frac{\partial}{\partial z}.
\end{equation}

We define now the vector bundle isomorphism
\begin{equation}\label{eq:flat_iso}
    \begin{aligned}
        {\flat} : TM &\to T^* M ,\\
         v &\mapsto \contr{v}  \dd \eta + \eta (v)  \eta.
    \end{aligned}
\end{equation}
We denote also by $\flat:\VecFields(M)\to \Forms^1(M)$ the associated isomorphism of $\Cont^\infty(M)$-modules. Notice that ${\flat}(\Reeb)=\eta$.

For a Hamiltonian function $H$ on $M$ we define the Hamiltonian vector field $X_H$ by
\begin{equation}\label{hamiltonian_vf_contact}
    {\flat} (X_H) = \dd H - (\Reeb (H) + H) \, \eta.
\end{equation}.

We call the triple $(M,\eta,H)$ a \emph{contact Hamiltonian system}. From \cref{hamiltonian_vf_contact} one can easily deduce that
\begin{subequations}\label{eqs:ham_vfs_characterization}
    \begin{align}
        \eta(X_H) &= -H, \\
        \lieD{X_H} \eta &= -\Reeb(H) \eta.\label{eq:ham_vf_conf_contactomorphism},
    \end{align}
\end{subequations}
which, can be proved equivalent to~\cref{hamiltonian_vf_contact}.

These two additional identities will be useful in what follows:
\begin{subequations}
    \begin{align}
        X_H(H) &= -\Reeb(H) H,\\
        \contr{X_H} \dd \eta &= \dd H - \Reeb(H) \eta.
    \end{align}
\end{subequations}

We say that a vector field $X\in \VecFields(M)$ is an \emph{infinitesimal conformal contactomorphism} for $(M,\eta)$ if $\lieD{X}\eta = a_X \eta$ for some function $a_X \in \Cont^\infty (M)$. When $a_X=0$, we say that $X$ is an \emph{infinitesimal contactomorphism}. Equivalently $X$ is a conformal contactomorphism if and only if its flow preserves the contact form $\eta$. $X$ is an {infinitesimal conformal contactomorphism} if and only if its flow preserves the contact distribution $\ker \eta$.

Note that, by~\cref{eq:ham_vf_conf_contactomorphism}, a Hamiltonian vector field is an infinitesimal conformal contactomorphism. Conversely, if $X$ is an infinitesimal conformal contactomorphism, then $X$ is the Hamiltonian vector field of $f = -\eta(X)$, and, moreover, $a_X = - \Reeb(f)$. Hence, in contact geometry Hamiltonian vector fields coincide with the infinitesimal conformal contactomorphism. Moreover $X_f$ is a conformal contactomorphism if and only if $\Reeb(f)=0$.

A contact manifold is an example of a Jacobi manifold~\cite{Lichnerowicz1978,Kirillov1976}, whose definition we recall below.

\subsection{Jacobi manifolds}
\begin{definition}\label{def:jacobi_mfd}
    A \emph{Jacobi manifold} is a triple $(M,\Lambda,E)$, where $\Lambda$ is a bivector field (a skew-symmetric contravariant 2-tensor field) and $E \in \VecFields (M)$ is a vector field, so that the following identities are satisfied:
    \begin{align}
        \lieBr{\Lambda,\Lambda} &= 2 E \wedge \Lambda \\
        \lieD{E} \Lambda &= \lieBr{E,\Lambda} = 0,
    \end{align}
    where $\lieBr{\cdot,\cdot}$ is the Schouten–Nijenhuis bracket~\cite{Lichnerowicz1978,Libermann1987,Vaisman1994}.
\end{definition}

The bivector $\Lambda$ induces a morphism of vector bundles
\begin{equation}
    \begin{aligned}
        \lsharp: T^*M &\to TM,\\
        \alpha &\mapsto \Lambda(\alpha,\cdot ).
    \end{aligned}
\end{equation}
We also denote by $\lsharp: \Forms^1(M) \to \VecFields(M)$ to the corresponding morphism of $\Cont^\infty(M)$-modules.

The Jacobi bracket associated to the Jacobi structure $(M,\Lambda,E)$ is given by
\begin{equation}
    \jacBr{f,g} = \Lambda(\dd f, \dd g) + f E(g) - g E(f).
\end{equation}

This bracket is bilinear, antisymmetric and satisfies the Jacobi identity, but it fails to satisfy the Leibniz rule; instead it satisfies the this weak version
\begin{equation}
    \supp(\{f,g\}) \subseteq \supp (f) \cap \supp (g).
\end{equation}

So, $(\Cont^\infty(M),\jacBr{\cdot,\cdot})$ is a local Lie algebra in the sense of Kirilov~\cite{Kirillov1976,Lichnerowicz1978}.

If $E=0$, then $(M,\Lambda)$ a Poisson manifold, and the bracket satisfies the Leibniz rule. This is the case of symplectic manifolds $(M,\omega)$, where $\Lambda$ is the contravariant inverse of the symplectic formula $\omega$. 

On a Jacobi manifold $(M,\Lambda, E)$ the Hamiltonian vector  $X_f$ associated to a function $f$ is given by
\begin{equation}
    X_f = \lsharp(\dd f) + f E,
\end{equation}
where $\lsharp(\alpha)(\beta) = \Lambda(\alpha,\beta)$, for arbitrary $1$-forms $\alpha,\beta$.

\subsection{The Jacobi structure of a contact manifold}

A contact manifold $(M,\eta)$ has a natural Jacobi structure $(M,\Lambda, E)$, where
\begin{equation}\label{eq:contact_jacobi}
    \Lambda(\alpha,\beta) = 
    -\dd \eta ({\flat}^{-1} (\alpha), {\flat}^{-1}(\beta)), \quad
    E = - \Reeb.
\end{equation}

A simple computation shows that
\begin{equation}
    \sharp(\alpha) = \lsharp(\alpha) - \alpha(\Reeb)\Reeb 
\end{equation}
for any $\alpha \in \Forms^1(M)$, since $\Reeb$ generates the kernel of $\dd \eta$ and, by duality, $\eta$ generates the kernel of $\Lambda$.

Given a function $f\in \Cont^\infty(M)$, one has that the Hamiltonian vector field $X_f$ is just
\begin{equation}
    X_f = \lsharp(\dd f) - f \Reeb.
\end{equation}

Furthermore, the Jacobi bracket of two functions $f,g$ on $(M,\eta)$ is given by
\begin{equation}
    \jacBr{f,g} = \Lambda(\dd f, \dd g) - f\Reeb(g) + g \Reeb(f).
\end{equation}

\begin{lemma}\label{lem:jacbr_1}
    For all $f,g\in \Cont^\infty(M)$ we have that
    \begin{equation}
        \jacBr{f,g} = X_f(g) + g \Reeb(f)  = -X_g(f) - f \Reeb(g)
    \end{equation}
\end{lemma}
\begin{proof}
    Indeed,
    \begin{align*}
        \jacBr{f,g} &=
         \Lambda(\dd f, \dd g) - f\Reeb(g) + g \Reeb(f)\\ &=
         X_f(g) + f \Reeb(g) - f \Reeb(g) + g \Reeb(f) \\ &=
         X_f(g) + g \Reeb(f),
    \end{align*}
    since
    \begin{equation*}
        X_f(g) = \dd g (\lsharp(\dd f)) - f \Reeb(g)
        = \Lambda(\dd f, \dd g) - f \Reeb(g).
    \end{equation*}

    The second equality follows from the antisymmetry of the bracket.
\end{proof}

\begin{lemma}\label{eq:jacBr_lieBr}
    For all $f,g\in \Cont^\infty(M)$ we have that
    \begin{equation}
        \jacBr{f,g} = -\eta(\lieBr{X_f,X_g})\,
    \end{equation}
\end{lemma}
\begin{proof}
    Using Cartan's formula
    \begin{align*}
        -\eta(\lieBr{X_f,X_g}) &=
        - X_f(\eta(X_g)) + (\lieD{X_f} \eta)(X_g),
        \intertext{by using~\eqref{eqs:ham_vfs_characterization},} &=
        - X_f(-g) + (- \Reeb(f) \eta)(X_g), 
        \intertext{using again~\eqref{eqs:ham_vfs_characterization},} &= 
        X_f(g) + g\Reeb(f) = \jacBr{f,g},
    \end{align*}
    by \cref{lem:jacbr_1}.
\end{proof}

Let us observe that, contrary to the case of the Poisson bracket, in our context if we have two functions in involution, say $ \jacBr{f,g} = 0$, this does not imply that $g$ is a constant of motion for $X_f$. 
However, if $\jacBr{f,g}=0$, then we get
\begin{equation}
    \jacBr{f,g} = X_f(g) + g \Reeb(f) =0,
\end{equation}
which implies
\begin{equation}
    X_f(g) = -\Reeb(f) g.
\end{equation}

Therefore, since for a given Hamiltonian we have $X_H(H) = - \Reeb(H)H$, if $f$ commutes with $H$ we obtain
\begin{equation}
    X_H(f) = - \Reeb(H) f.
\end{equation}

Because of this, make the following definition
\begin{definition}
    In a Hamiltonian system $(M, \eta, H)$, we say that a function $f$ is \emph{dissipated} if $\jacBr{H,f} = 0$. Equivalently, $f$ dissipates at the same rate as the Hamiltonian.
\end{definition}

We note that the set of dissipated functions is a Lie subalgebra of $(\Cont^\infty(M),\jacBr{\, \cdot \, , \, \cdot \,})$. Indeed, $\RR$-linear combinations of dissipated functions are dissipated, and, because of the Jacobi identity, the Jacobi bracket of two dissipated functions is dissipated. Moreover, it is an algebra over the set of conserved quantities; this is, if $f$ is a dissipated quantity and $g$ is a conserved quantity, then $fg$ is dissipated:
\begin{equation}
    X_H (f g) = g X_H(f) = - \Reeb(H) f g.
\end{equation}

If we assume that $H$ has no zeros, we can relate dissipated functions to conserved functions. Assume that $f$ is dissipated, then $f/H$ is a conserved quantity. Indeed:
\begin{equation}
    X_H \parens*{\frac{f}{H}} = \frac{X_H(f) H - f X_H(H) }{H^2} = 
    \frac{-\Reeb(H) f H  + \Reeb(H) f H}{H^2} = 0.
\end{equation}

In general, if $f_1,f_2$ commutes with $H$, then $f_1/f_2$ is a conserved quantity, assuming $f_2$ has no zeros.

The conclusion is that in order to obtain conserved quantities one should find quantities that dissipate at the same rate as the Hamiltonian.

\begin{remark}
  In the particular case where $\Reeb(H)=0$, then the dissipated quantities are precisely the conserved quantities. That is, $\jacBr{H,f}=0$ if and only if $X_H(f)=0$. 
\end{remark}

In the case that $H$ has no zeros there is a correspondence between sets of $m$ independent conserved quantities and sets of $m$ dissipated quantities by taking the quotients. Explicitly, if $f_1,\ldots 
f_m$ commute with $H$, then
\begin{equation}
    g_i = \frac{f_i}{H}
\end{equation}
are conserved quantities. Conversely, if $g_1,\ldots g_m$ are conserved quantities, then
\begin{equation}
    f_i = {g_i}{H}
\end{equation} 
are dissipated quantities.

\subsection{Infinitesimal symmetries for a contact Hamiltonian system}
Next, we will introduce a class of infinitesimal symmetries for a contact Hamiltonian system $(M,\eta,H)$ which will be very useful on the next section. First we prove the following result, which help us to compute Jacobi brackets.
\begin{proposition}\label{prop:jacBr_alternative}
    Let $X$ be a vector field such that $\eta(X) = - f$, then
\begin{equation}
    \jacBr{H,f} = 
    -\eta{(\lieBr{X_H,X})} = (\lieD{X} \eta) (X_H) + X(H).
\end{equation}
\end{proposition}

\begin{proof}
    If $\eta(X) = - f$, then $\eta(X - X_f) = 0$, so that $X-X_f$ is in the kernel of $\eta$.
    
    Since
    \begin{equation*}
        \lieD{X} \eta = - {\Reeb}(H) \eta,
    \end{equation*}
    we deduce that
    \begin{equation*}
        (\lieD{X} \, \eta)(X_f) = (\lieD{X} \, \eta)(X).
    \end{equation*}

    Therefore, using~\cref{eq:jacBr_lieBr} and Cartan's formula twice, one finds
\begin{align*}
\{H, f\} &=  - \eta ([X_H, X_f ] )\\
& =  (\lieD{X_H} \, \eta) (X_f) - X_H (\eta(X_f))\\
&= (\lieD{X_H} \eta)(X) - X_H(\eta(X)) \\
&= - \eta ([X_H, X]).
\end{align*}
From the second equality, we have
\begin{align*}
 - \eta ([X_H, X] )  &=  (\lieD{X} \, \eta) (X_H) - X (\eta(X_H))\\
&= (\lieD{X} \eta)(X_H) + X(H),
\end{align*}
applying again Cartan's formula.
\end{proof}

The above Proposition suggests us to introduce the following definition.

\begin{definition}
    A \emph{dynamical symmetry} on a contact Hamiltonian system $(M,\eta,H)$ is a vector field $X$ such that $\eta(\lieBr{X_H,X})=0$.
\end{definition}

Using \cref{prop:jacBr_alternative}, we deduce the following.
\begin{theorem}\label{thm:dynamical_symmetry}
    Let $X$ be a vector field on $M$  Then, $X$ is a dynamical symmetry of $(M,\eta,H)$ if and only if $\eta(X)$ is a dissipated quantity. 
\end{theorem}

\begin{remark}
    The natural correspondence between dynamical symmetries and dissipated  quantities $f$ is not one-to-one. Indeed, given a dynamical symmetry $X$ such that $-\eta(X)=f$, the set of vector fields $\mathcal{F}=\set{X_f + Y \mid Y \in \ker \eta}$ are the dynamical symmetries corresponding to the quantity $f$. As one easily sees from~\cref{eqs:ham_vfs_characterization}, $X_f$ is the only one which is a Hamiltonian vector field.
\end{remark}

There is another concept of symmetry on this setting: \emph{Cartan symmetries}.

\begin{definition}
    We say that $X \in \VecFields(M)$ is a \emph{Cartan symmetry} for $(M,\eta,H)$ if $\lieD{X} \eta =  a \eta + \dd g$ for some functions $a,g\in \Cont^{\infty}(M)$ and $X(H) =  a H + g \Reeb(H)$. 
\end{definition}
\begin{theorem}\label{thm:Cartan}
    Let $X$ be a Cartan symmetry such that $\lieD{X}\eta  = \dd g + a \eta$. Then $f=\eta(X) - g$ is a dissipated quantity. 
\end{theorem}

\begin{proof}
    From \cref{prop:jacBr_alternative}, we have
    \begin{align*}
        \jacBr{H,f} &=\jacBr{H,\eta(X)}  - \jacBr{H,g} \\ &= 
        (\lieD{X}\eta) (X_H) + X(H) - X_H (g) - g \Reeb(H) \\ &= 
        a\eta (X_H) - \dd g (X_H) + X(H) - X_H (g) - g \Reeb (H) \\ &=
        -a H  + X(H) - g \Reeb(H) = 0.
    \end{align*}
\end{proof}

\begin{remark}\label{rem:dynamical_cartan}
    A Cartan symmetry such that $\lieD{X}\eta  = \dd g + a \eta$ is a dynamical symmetry  when $\dd g =0$.

    A dynamical symmetry $X$ is a Cartan symmetry when $X$ is a Hamiltonian vector field.
\end{remark}

\section{Symmetries and conserved quantities in contact Lagrangian systems}
We consider a contact system given by a regular Lagrangian $L:TQ\times \RR \to \RR$ and the contact Lagrangian form 
\begin{equation}
    \eta_L = \dd z - \alpha_L,
\end{equation}
where
\begin{align}
    \alpha_L &= S^* (\dd L) = \pdv{L}{\dot q^i} \dd {q}^i,
\end{align}
 where $S$ is the canonical vertical endomorphism on $TQ$ extended in the natural way to $TQ \times \RR$, and $(q^i, \dot{q}^i,z)$ denote bundle coordinates on $TQ \times \RR$, and $z$ is a global coordinate in $\RR$.
 
 The energy of the system is defined by
 \begin{equation}
    E_L = \Delta(L) - L = \dot{q}^i \pdv{L}{\dot{q}^i} - L,
 \end{equation}
 where $\Delta$ is the Liouville vector field on $TQ$ extended to $TQ\times \Reeb$ in the natural way.

 So, $(TQ\times \RR, \eta_L, E_L)$ is a contact Hamiltonian system in the sense discussed in Section~2.

The Reeb vector field, denoted by $\Reeb_L$ is given by
\begin{equation}
    \Reeb_L = \frac{\partial}{\partial z} - 
    W^{ij} \frac{\partial^2 L}{\partial \dot{q}^i \partial z} \pdv{}{\dot{q}^j},
\end{equation}
where $(W^{ij})$ is the inverse of the Hessian matrix with respect to the velocities
\begin{equation}
    (W_{ij}) = \parens*{\frac{\partial^2 L}{\partial \dot{q}^i \partial \dot{q}^j}}.
\end{equation} 

The Hamiltonian vector field of the energy will be denoted $\xi_L = X_{E_L}$, hence
\begin{equation}
    \flat_L(\xi_L) = \dd E_L - (\Reeb(E_L)+E_L)\eta_L,
\end{equation}
where $\flat_L(v) = \contr{v} \dd \eta_L + \eta_L (v) \eta_L$ is the isomorphism defined in \cref{eq:flat_iso} for this particular contact structure.

$\xi_L$ is a second order differential equation (SODE) (that is, $S(\xi_L) = \Delta$) and its solutions are just the ones of the generalized Euler-Lagrange equations for $L$:
\begin{equation}
    \dv{}{t} \parens*{\pdv{L}{\dot{q}^i}} - \pdv{L}{q^i} = \pdv{L}{\dot{q}^i} \pdv{L}{z}.
\end{equation}

A direct computation shows that
\begin{align}
    \Reeb_L (E_L) &= - \pdv{L}{z},\\
    dz(\xi_L) &= L.
\end{align}

In~\cite{deLeon2019} one can find a more complete exposition of the theory of contact Lagrangian systems.

\subsection[Lifts of vector fields on on Q and Q x R]{Lifts of vector fields on on $Q$ and $Q \times \RR$}
The vector bundle structure of $\tau_Q:TQ \to Q$ provides a rich geometry that we will exploited for our interests.

Let us recall some definition of lifts of vector fields on $Q$ to its tangent bundle.

Let $Y$ be a vector field on $Q$ given locally by
\begin{equation}
    Y = Y^i \pdv{}{q^i}.
\end{equation}
Its vertical lift is given by
\begin{equation}
    \vlift{Y} = Y^i \pdv{}{\dot{q}^i}.
\end{equation}
The complete lift is given by
\begin{equation}
    \clift{Y} = Y^i \pdv{}{q^i} + \dot{q}^j \pdv{Y^i}{{q}^j} \pdv{}{\dot{q}^i}.
\end{equation}
These lifts can be defined geometrically~\cite{deLeon2011,Yano1973}.

The natural extensions of $\vlift{Y}$ and $\clift{Y}$ to $TQ \times \RR$ will be denoted by the same symbols.

One can consider more general vector fields. Indeed, let $Y$ be a vector field on $Q \times \RR$ given by
\begin{equation}
    Y = Y^i \frac{\partial}{\partial q^i} + \mathcal{Z} \frac{\partial}{\partial z},
\end{equation}
then its complete lift to $T(Q \times \mathbb{R})$ is
\begin{align*}
Y^C =& Y^i \frac{\partial}{\partial q^i} + \mathcal{Z} \frac{\partial}{\partial z}
+ \dot{q}^j \frac{\partial Y^i}{\partial q^j} \frac{\partial}{\partial \dot{q}^i} \\
& + \dot{q}^j \frac{\partial \mathcal{Z}}{\partial q^j} \frac{\partial}{\partial \dot{z}}
+ \dot{z} \frac{\partial Y^i}{\partial z} \frac{\partial}{\partial \dot{q}^i} + \dot{z} \frac{\partial \mathcal{Z}}{\partial z} \frac{\partial}{\partial \dot{z}},
\end{align*}
where $(z, \dot{z})$ are the bundle coordinates in $T\mathbb{R}  \cong \RR \times \RR $.

But we are restricted to vector fields which are tangent to the submanifold $TQ \times \mathbb{R}$ of $T(Q \times \mathbb{R})$ which is given by the equation
\begin{equation}
    \dot{z}=0.
\end{equation}

We shall only consider those vector fields $Y$ on $Q \times \mathbb{R}$ such that their complete lifts to
$T(Q \times \mathbb{R})$ 
are tangent to $TQ \times \mathbb{R}$. This just happens when
$$
\frac{\partial \mathcal{Z}}{\partial q^i} = 0,
$$
that is, $\mathcal{Z}$ does not depend on the positions $q$. The restriction of the complete lift  $Y^C$ to $TQ \times \mathbb{R}$ will be denoted by
\begin{equation}
    \rclift{Y} = Y^i \frac{\partial}{\partial q^i} + \mathcal{Z} \frac{\partial}{\partial z}
    + \dot{q}^j \frac{\partial Y^i}{\partial q^j} \frac{\partial}{\partial \dot{q}^i}
\end{equation}

In such a case, we will denote by $\rvlift{Y}$ the vertical lift of the projection of $Y$ to $Q$, say
$$
\rvlift{Y} =  Y^i \frac{\partial}{\partial \dot{q}^i}
$$
which is obviously tangent to $TQ \times \mathbb{R}$.

We denote by $\rVecFields(Q\times \RR)$ to the set of vector fields on $Q\times\RR$ such that their complete lifts are tangent to $TQ\times\RR$.

\subsection{Infinitesimal symmetries}\label{sec:infinitesimal_symmetries}
We will now try to understand the symmetries of a Lagrangian system which are lifts of vector fields on the configuration space.

\begin{definition}
    We say that a vector field $Y \in \VecFields(Q)$ is an \emph{infinitesimal symmetry of $L$} if $\clift{Y}(L) = 0$.
\end{definition}

\begin{theorem}\label{thm:base_symmetry}
    Let $(M,\eta_L,E_L)$ be a contact Lagrangian system and let $Y \in \VecFields(Q)$. Then, $Y$ is an infinitesimal symmetry of $L$ if and only if $f=\vlift{Y}(L)$ is a dissipated quantity, that is, it commutes with $E_L$, or,
    \begin{equation}\label{eq:base_field_dissipation}
        \xi_L (f) = -\Reeb_L (E_L) f = \pdv{L}{z} f. 
    \end{equation}
\end{theorem}

For the proof of \cref{thm:base_symmetry}, we will use the identities listed on the following lemma, that can be proved by a direct computation. 
\begin{lemma}\label{Qcomputations}
    Let $Y\in \VecFields(Q)$. The following identities hold:
    \begin{align}
        \eta_L(\clift{Y}) &= - \vlift{Y}(L) \label{eq:base_contr_eta}, \\ 
        \lieD{\clift{Y}}\eta_L &= -\alpha_{\clift{Y}(L)} = - \pdv{\clift{Y}(L)}{\dot{q}^i} \dd q^i,\label{eq:base_lieD_eta} \\
        \Reeb_L (f) &= 0,\\
        \lieBr{\clift{Y},\Delta} &= \lieBr{\clift{Y}, S} = 0,\\
        S(\clift{Y}) &= \vlift{Y}
    \end{align}
\end{lemma}

\begin{proof}[Proof of \cref{thm:base_symmetry}]
    We can now compute the Jacobi brackets using the identities form the previous lemma and \cref{prop:jacBr_alternative}. Let $f = \vlift{Y}(L)$ so that $\eta(\clift{Y})=-f$. Then we have,
    \begin{align*}
        \jacBr{E_L,f} &= 
        \lieD{\clift{Y}} \eta_L (\xi_L) + \clift{Y}(E_L) \\ &=
        - \alpha_{\clift{Y}(L)} (\xi_L) + \clift{Y}(\Delta(L) - L) \\&=
         (S(\xi_L))(\clift{Y}(L)) -  \clift{Y}(\Delta(L) - L) \\ &=
         \Delta(\clift{Y}(L)) - \clift{Y}(\Delta(L)) - \clift{Y}(L) =
         - \clift{Y}(L).
    \end{align*}
    Therefore, the result follows.
\end{proof}

\begin{remark}
    We notice that whenever $Y$ is an infinitesimal symmetry of $L$, then $\clift{Y}$ is the Hamiltonian vector field of $\vlift{Y}(L)$.
\end{remark}

\begin{remark}
    This result should be compared with the \emph{First Noether Theorem} from~\cite{Georgieva2003} in the case that $L$ does not depend explicitly on time. The conserved quantity obtained from a symmetry on~\cite{Georgieva2003} is not a function on $M$, but a functional that depends on the chosen integral curve $\gamma$ of $\xi_L$. We can recover the result by noticing that the conserved quantity in~\cite{Georgieva2003} is given by
    \begin{equation}
        \begin{aligned}
            G[\gamma](t) &=  \exp\parens*{- \int_0^t \pdv{L}{z} \dd \gamma} (\vlift{Y}(L) \comp \gamma)(t) \\&=
            \exp \parens*{- \int_0^t \frac{\dot{f}}{f} \dd \gamma} (f\comp \gamma)(t),
        \end{aligned}
    \end{equation}
    hence, along an integral curve $\gamma$ of $\xi_L$, we have:
    \begin{equation}
        \dot{G}[\gamma](t) = -\dot{f} \exp \parens*{- \int_0^t \frac{\dot{f}}{f} \dd \gamma} + \exp \parens*{- \int_0^t \frac{\dot{f}}{f} \dd \gamma} \dot{f} = 0.
    \end{equation}
\end{remark}

We now consider vector fields $Y \in \rVecFields(Q \times \Reeb)$  of the form
\begin{equation}
    Y = Y^i \pdv{}{q^i} + \mathcal{Z} \pdv{}{z}.
\end{equation}

To go further in this case, we need to extend the computations in \cref{Qcomputations}.

\begin{lemma}\label{QRcomputations}
    Let $Y\in \rVecFields(Q\times \RR)$ with $z$ component $\mathcal{Z}$. Let $f=-\eta_L(\rclift{Y})$. 
    Then, we have:
    \begin{align}
        \eta_L(\rclift{Y}) &= -(\rvlift{Y}(L) - \mathcal{Z}) \label{eq:base_ext_contr_eta}, \\ 
        \lieD{\rclift{Y}}\eta_L &= -\Reeb_L(f) \dd z -\alpha_{\rclift{Y}(L)},\label{eq:base_ext_lieD_eta} \\
        S(\rclift{Y}) &= \rvlift{Y}.
    \end{align}
\end{lemma}

In this setting, we can provide generalization of the concept of infinitesimal symmetry and the corresponding dissipated quantities.
\begin{definition}
    We say that a vector field $Y \in \rVecFields(Q\times \RR)$ is an \emph{generalized infinitesimal symmetry of $L$} if $\clift{Y}(L) = -\Reeb(f) L$,
    where $f = -\eta_L(\rclift{Y})$
\end{definition}

\begin{theorem}\label{thm:base_ext_symmetry}
    Let $Y \in \VecFields (Q\times \RR)$. Then $f = {\rvlift{Y}(L) - \mathcal{Z}}$ is a dissipated quantity if and only if $Y$ is a generalized infinitesimal symmetry.
\end{theorem}

\begin{proof}[Proof of \cref{thm:base_ext_symmetry}]
    We proceed as in \cref{thm:base_symmetry}. Let $f = \rvlift{Y}(L) - \mathcal{Z}= - \eta(\rclift{Y})$.
    \begin{align*}
        \jacBr{E_L,f} &= 
        (\lieD{\rclift{Y}} \eta_L) (\xi_L) + \rclift{Y}(E_L) \\ &=
        -\Reeb_L(f)\dd z(\xi_L) - \alpha_{\rclift{Y}(L)} (\xi_L) + \rclift{Y}(\Delta(L) - L) \\&=
         -\Reeb_L(f)L +(S(\xi_L))(\rclift{Y}(L)) -  \rclift{Y}(\Delta(L) - L) \\ &=
         -\Reeb_L(f)L+\Delta(\rclift{Y}(L)) - \rclift{Y}(\Delta(L)) - \rclift{Y}(L) \\ &=
         -\Reeb_L(f)L - \rclift{Y}(L).
    \end{align*}
    Therefore, the result follows.
\end{proof}

\begin{remark}
    In this case, the fact that $\rclift{Y}(L) = - \Reeb_L(f)L$ does not ensure that $\rclift{Y}$ is a Hamiltonian vector field. Indeed, from \cref{eq:base_ext_lieD_eta}, we compute $\lieD{\rclift{Y}} \eta_L = -\Reeb_L(f)\eta_L + L \alpha_{\Reeb_L(f)}$, hence $\rclift{Y}$ is Hamiltonian if and only if $ \alpha_{\Reeb_L(f)} L = 0$.
\end{remark}

\subsection{Noether symmetries}
\begin{definition}
    We say that $Y \in \rVecFields (Q\times \RR)$ with $z$ component $\mathcal{Z}$ is a \emph{Noether symmetry} if $\rclift{Y}$ is a Cartan symmetry.
\end{definition}

From the conservation theorem for Cartan symmetries (\cref{thm:Cartan}), we can deduce a new one for Noether symmetries.

\begin{theorem}\label{noether}
    Let $Y$ be a Noether symmetry such that $\lieD{\rclift{Y}}(\eta_L) = \dd g + a \eta_L$. Then $f=\rvlift{Y}(L) - \mathcal{Z} - {g}$ is a dissipated quantities. 
\end{theorem}
\begin{proof}
    Apply \cref{thm:Cartan} to $\rclift{Y}$, so, using \cref{QRcomputations}, $\eta(\rclift{Y})-g= \rvlift{Y}(L) - \mathcal{Z} - {g} =f$ commutes with $E_L$.
\end{proof}

\begin{remark}
    Infinitesimal symmetries on $Q$ are Noether symmetries for $g=0$. However, general infinitesimal symmetries on $Q\times \RR$ can fail to be Noether symmetries because its complete lift is not a Hamiltonian vector field. See~\cref{rem:dynamical_cartan}
 \end{remark}

\subsection{Lie symmetries}
\begin{definition}
    A Lie symmetry is a vector field $Y \in \rVecFields (Q\times \RR)$ such that $\rclift{Y}$ is a dynamical symmetry.
\end{definition}
As a consequence of \cref{thm:dynamical_symmetry}, 
\begin{theorem}
    Let $Y$ be a Lie symmetry. Then $f=-\eta_L(\rclift{Y}) = \rvlift{Y}(L)- \mathcal{Z}$ is a dissipated quantity.
\end{theorem}

\begin{remark}
    Any infinitesimal symmetry of the Lagrangian $L$ is a Lie symmetry.
\end{remark}

\section{Lie group of symmetries on a contact Hamiltonian system}
An important case of symmetries for contact Hamiltonian or Lagrangian systems appears when a Lie group preserving the geometric structure and the energy.

As it is well-known, if $G$ is a Lie group acting on a contact manifold $(M,\eta)$ by contactomorphisms, then there exist a momentum map
\begin{equation}
    J: M \to \mathfrak{g}^*,
\end{equation}
where $\mathfrak{g}^*$ is the dual of the Lie algebra $\mathfrak{g}$ of $G$, and $J$ is defined by
\begin{equation*}
    J(x)(\xi) = - \eta_x(\xi_M(x)),
\end{equation*}
where $\xi_M \in \VecFields(M)$ is the infinitesimal generator of the flow $\Phi$ of $\xi$:
\begin{equation}
    \phi_t(x) = (\exp (t \xi)) x
\end{equation}
We denote by $\hat{J}: \mathfrak{g} \to \Cont^\infty(M)$ to the map
\begin{equation}
    \hat{J}(\xi)(x) = J(x)(\xi).
\end{equation}

Since the action 
\begin{equation}
    \begin{aligned}
        \Phi_g : M &\to M,
        x \mapsto \Phi_g (x) = g \cdot x
    \end{aligned}
\end{equation}
preserves $\eta$ (i.e., $\Phi_g^* \eta = \eta$), we can prove that $J$ is indeed a momentum map, and, in addition, it is $\Ad^*$-equivariant. Furthermore, we deduce that $\Reeb(\hat{J}(\xi))=0$ for all $\xi \in \mathfrak{g}$.

Assume now that $H$ is a Hamiltonian function defined on $M$ that is invariant by the action, that is,
\begin{equation}
    H \comp \Phi_g = H, \, \forall g \in G.
\end{equation}
Then we deduce that
\begin{equation}
    \lieD{\xi_M} H =0.
\end{equation}

Using \cref{prop:jacBr_alternative}, we can deduce that
\begin{equation}
    \jacBr{H,\hat{J}(\xi)} = X_H (\hat{J}(\xi)) + (\hat{J}(\xi)) \Reeb(H).
\end{equation}
But,
\begin{align*}
    \jacBr{H, \hat{J}(\xi)} &= - \jacBr{\hat{J}(\xi),H} \\ &=
    -X_{\hat{J}(\xi_M)} - H \Reeb(\hat{J}(\xi)) \\ &= 
    -\xi_M(H) - H \Reeb(\hat{J}(\xi)) = 0.
\end{align*}

Therefore,
\begin{equation}
    X_H(\hat{J}(\xi) ) = -\Reeb(H) \hat{J}(\xi).
\end{equation}
That is, $\hat{J}(\xi)$ is a dissipated quantity. Therefore, we have obtained the following:
\begin{theorem}
    $\xi_M$ is a dynamical symmetry for $(M,\eta,\xi)$ and $\hat{J}(\xi)$ is a dissipated function.
\end{theorem}

\section{Lie groups acting on contact Lagrangian systems}
Assume that a Lie group $G$ acts on $Q$
\begin{equation}
    \Phi: G\times Q \to Q,
\end{equation}
such that the action preserves a (regular) Lagrangian $L:TQ\times \RR \to \RR$. This means that the lifted action to $TQ \times \RR$,
\begin{equation}
    \tilde{\Phi}: G\times TQ \to TQ,
\end{equation}
given by $\tilde{\Phi} = (T\Phi, \Id)$ preserves $L$. As a direct consequence, $G$ preserves the contact form $\eta_L$. In other words, $G$ acts by contactomorphisms on $(TQ\times \RR, \eta_L)$.

Consider the corresponding momentum maps:
\begin{equation}
    \begin{aligned}
        J_L: TQ \times \RR &\to \mathfrak{g}^*,\\
        J_L(v_q,z)(v_q,\xi) & = - \eta_L(\xi_{TQ \times \RR}).   
    \end{aligned}
\end{equation}

Notice that 
\begin{equation}
    \xi_{TQ \times \RR} = \clift{\xi_Q} + \pdv{}{z}.
\end{equation}

Using the results of \cref{sec:infinitesimal_symmetries}, we conclude that $\xi_Q$ is an infinitesimal symmetry of $L$ and the function
\begin{equation}
    f =  \vlift{\xi_Q}(L)
\end{equation}
is a dissipated quantity.

\section*{Acknowledgements}
This  work  has  been  partially  supported  by the MINECO  Grants  MTM2016-76-072-P and the ICMAT Severo Ochoa projects SEV-2011-0087 and SEV-2015-0554.  Manuel Lainz wishes to thank MICINN and ICMAT for a FPI-Severo Ochoa predoctoral contract PRE2018-083203.

\printbibliography
\end{document}